\RequirePackage{etoolbox}
\csdef{input@path}{%
 {sty/}
 {img/}
}%
\csgdef{bibdir}{bib/}

\documentclass[ba]{imsart}
\pubyear{0000}
\volume{00}
\issue{0}
\doi{0000}
\firstpage{1}
\lastpage{1}

\usepackage{amsmath,amsfonts,amsthm, amssymb}
\usepackage[colorlinks,citecolor=blue,urlcolor=blue,filecolor=blue,backref=page]{hyperref}
\usepackage{graphicx}
\usepackage[linesnumbered,ruled,vlined]{algorithm2e}
\usepackage{paralist,booktabs}

\startlocaldefs
\theoremstyle{plain}
\newtheorem{thm}{Theorem}[section]

\newtheorem{cor}[thm]{Corollary}

\newtheorem{theorem}{Theorem}[section]
\newtheorem{lemma}{Lemma}[section]

\newtheorem{corollary}{Corollary}[section]

\newtheorem{example}{Example}[section]
\newcommand{\prf}{\begin{proof}{Proof}}

 \def\a{{\alpha}}
 \def\b{{\beta}}                         
\def\C{{\cal C}}                                        
 \def\d{{\delta}}      
                 \def\e{{\varepsilon}}                   
                          
 \def\g{{\gamma}}                                                  
 \def\h{{\eta}}

 \def\l{{\lambda}}     
                           
                 \def\n{{\nu}}                                                        
 \def\o{{\omega}}
                           
 \def\q{{\theta}}                       
\def\R{{\cal R}} 				\def\RR{\mathbb{R}}
                 \def\s{{\sigma}}       
                 \def\t{{\tau}}

\def\ra{\rightarrow}

\def\dotfil{\leaders\hbox to 1em{\hss.\hss}\hfill}

\def\hexnumber#1{\ifcase#1 0\or 1\or 2\or 3\or 4\or 5\or 6\or 7\or 8\or 9\or A\or B\or C\or D\or E\or F\fi}

\def\BL{{\rm B\kern -0.5pt L}}
\def\UC{{\rm U\kern-0.5pt C}}

\def\argmax{\mathop{\rm argmax}}

\def\implies{\hbox{ $\Rightarrow$ }}

\def\qed{\ \vrule height4pt width4pt depth0pt}

\def\tdot#1{\kern1.2pt\dot{\vphantom{#1}}\kern-1.2pt
        \dot#1\kern0.8pt\dot{\vphantom{#1}}\kern-0.8pt}

\def\E{\mathord{\rm E}}
\def\Pr{\mathord{\rm P}}
\def\var{\mathop{\rm var}\nolimits}

\mathchardef\given="626A
\mathcode`:="603A

\def\generalweak#1{\ {\mathchoice{\buildrel #1\over \rightsquigarrow}%
{\raise-2pt\hbox{$\buildrel #1\over \rightsquigarrow$}}{}{}}\ }

\def\generalprob#1{\ {\mathchoice{\raise-1.5pt\hbox{$\buildrel#1\over\ra$}}
{\raise-2pt\hbox{$\buildrel#1\over\ra$}}{}{}}\ }

\def\prob{\generalprob{\small P}}

\def\Pgg{\ {\mathchoice{\raise-1.5pt\hbox{$\buildrel {\small P}\over\gg$}}
{\raise-2pt\hbox{$\buildrel {\small P}\over\gg$}}{}{}}\ }

\def\boxit#1{\vbox{\hrule\hbox{\vrule\kern1pt\vbox{\kern1pt #1\kern1pt}\kern1pt\vrule}\hrule}}
\def\eet#1{}

\def\th{\theta}

\newcommand{\raff}[1]{\renewcommand{\arraystretch}{#1}}
\endlocaldefs

\begin{document}

\begin{frontmatter}
\title{Bayesian tolerance regions, with an application to linear mixed models}
\runtitle{Bayesian tolerance regions}

\begin{aug}
\author{\fnms{X. Gregory} \snm{Chen\thanksref{t1}} \ead[label=e1]{xiangyi.gregory.chen@msd.com}}
\and \author{\fnms{A.W.} \snm{van der Vaart\thanksref{t2}}\ead[label=e2]{a.w.vandervaart@tudelft.nl}}
\thankstext{t1}{Biostatistics and Research Decision Sciences, MSD, }
\thankstext{t2}{DIAM, Delft Institute of Applied Mathematics. The research leading to these results has received funding from the European Research Council under ERC Grant Agreement 320637 and 
is partly financed by the NWO Spinoza prize awarded to  A.W. van der Vaart
 by the Netherlands Organisation for Scientific Research (NWO).}

\runauthor{X. Chen and A.W. van der Vaart}

\end{aug}

\begin{abstract}
We review and contrast frequentist and Bayesian definitions of tolerance regions.
We give conditions under which for large samples a Bayesian region also has frequentist validity,
and study the latter for smaller samples in a simulation study. 
We discuss a computational strategy for computing a Bayesian two-sided
tolerance interval for a  Gaussian future variable, and
apply this to the case of possibly unbalanced linear mixed models.
We illustrate the method on a quality control experiment from the pharmaceutical industry.
\end{abstract}


\end{frontmatter}


\section{Introduction}
The concept of \emph{tolerance region} is of central importance in quality control. A tolerance region is 
a prediction set for a future observation,  which takes account of both
the random nature of this observation and the uncertainty about its distribution. 
It incorporates the statistical error of estimation of unknown parameters of this distribution.

Conventional tolerance regions take the uncertainty of estimated parameters into account in one of two ways.
Either the region captures the future observation a fraction of $1-\a$ times on average over both future and past
observations (the \emph{$(1-\a)$-expectation tolerance region}), 
or the region captures the future observation with probability at least $1-\d$ with $1-\a$ confidence
over past observations (the \emph{$(\d,\a)$-tolerance region}).  (We give precise definitions
in Section~\ref{SectionDefinitions}.) 
The second way appears to be preferred in the pharmaceutical industry.
The ``on average" and ``confidence" can refer to the sampling distribution of the data in a frequentist  sense,
but can also refer to a posterior distribution in the Bayesian statistical framework.
The main focus on the present paper is the second, but we do relate it to the frequentist setup.


Frequentist tolerance regions have been well studied in the literature.  A general reference is the book
\cite{krishnamoorthy2009}, especially for the situation that the data are i.i.d.  For the linear mixed model (LMM), the paper
\cite{Gaurav2012} provides an elegant solution to build one- and two-sided $(\d,\a)$-tolerance intervals, 
and includes a comprehensive review of the literature. One purpose of the present paper is to
provide a Bayesian approach for the general LMM.

The Bayesian formulation of tolerance regions dates back to at least 1964, 
but the subsequent literature is relatively small.  In his paper \cite{AitchisonBTR}, Aitchison derived 
Bayesian $(\d,\a)$-tolerance regions from a decision-theoretic framework, and contrasted them to the
frequentist counterparts.  In \cite{AitchisonBTR2} he extended his discussion to $(1-\a)$-tolerance
regions, which are a natural way to build Bayesian prediction intervals.  

A one-sided $(\d,\a)$-tolerance interval for a univariate future observation is usually easy to
compute, but two-sided tolerance intervals pose challenges, both conceptually and computationally.
There are at least two common approaches: intersecting two one-sided tolerance intervals, or fixing
one degree of freedom of the interval (e.g.\ the midpoint of the interval).  The former approach is
identical to specifying probability masses in the two tails of the distribution of the future
variable separately and gives a valid construction, in view of Bonferroni's inequality, but it yields
longer intervals than necessary.  (They are called ``anti-conservative'' in the pharmaceutical industry in reference to
the customers, whereas statisticians use the term ``conservative''.) See \cite{Hamada2002} for an example
application.  The second approach, fixing one degree of freedom, is the conventional choice, especially in the frequentist
framework, but requires untangling the dependence of the interval on the unknown true parameter.
Solutions are often not available in analytical form and computationally more challenging.
Wolfinger in \cite{wolfinger} proposed an algorithm to derive a two-sided Bayesian interval for a future
normal variate, which was refined by Krishnamoorthy and Mathew \cite{krishnamoorthy2009}.
Their algorithms have been widely adopted in practice, and also in other literature
(e.g.\cite{Katki2005}, \cite{Merwe2006}, \cite{Merwe2007}). 

Other contributions to the literature include the following:
\cite{Miller1989} used the empirical Bayes method to construct a one-sided tolerance
interval given an i.i.d. sample from a normal distribution;
\cite{Hamada2004} derived a Bayesian tolerance interval that contains a proportion of observations with a specified confidence;
\cite{Easterling1970} and \cite{Young2016} focused on the sample size needed to attain a certain accuracy;
\cite{Merwe2006} and \cite{Merwe2007} allowed data from the unbalanced one-way random effects model
and the balanced two-factor nested random effects model;
\cite{Mukerjee2001} discussed probability matching priors (PMP) in the one-sided case to ensure second-order frequentist
validity;
\cite{DPRMSHong2014} extended this to the two-sided case;
\cite{DPSHOng2014} incorporated it to a balanced one-way random effects model, 
and evaluated its performance against the frequentist method MLS in \cite{krishnamoorthy2009}.

Although the PMP approach has merit when the sample size is
small, it is analytically demanding even when data are i.i.d., and it seems difficult to extend to
the general LMM setting. The algorithms of Wolfinger \cite{wolfinger} and Krishnamoorthy and Mathew \cite{krishnamoorthy2009} 
can be extended to LMM, but they oversimplify the target function during optimization 
and may result in less satisfactory performance.

In this paper we propose a computationally efficient solution for the general case that the 
future observation possesses a normal distribution. We show that this
is easy to implement given any data model for which a sample from the posterior distribution is available. 
We investigate when the shortest interval is centered at the posterior mean of the parameter.
We discuss the interval in particular for the linear mixed model, and within this context show its good performance by simulation. 
We illustrate the method on an  example that is representative for pharmaceutical applications.
Finally we also prove that the Bayesian interval has frequentist validity in the case of large samples.

\section{Definitions and setup}
\label{SectionDefinitions}
Given are observed data $X$, with a distribution $P_\q$ depending on a parameter $\q$, 
and future unobserved ``data'' $Z$ , with a distribution $Q_\q$ depending on the
same parameter $\q$. In both cases the sample space is arbitrary.
A tolerance region is a set $\R(X)$ in the sample space of $Z$
that captures $Z$ with a ``prescribed probability''. It will typically 
be constructed using the observation $X$ to overcome the problem that  $\q$, and
hence the law of $Z$, is unknown.
There are various ways to make the ``prescribed probability'' precise, and these can be
divided into frequentist and Bayesian definitions. The probability statement
will refer to both $X$ and $Z$, and is fixed by one or two parameters $\a$ and $\d$, 
which are typically chosen small, e.g. 5\%. 

The parameter $\q$ will typically be chosen to identify the distribution of $Z$. The distribution
of $X$ may also depend on unknown ``nuisance'' parameters. 
For simplicity of notation we do not make this explicit in the following.
We shall use the notation $\Pr$ or $\Pr_\q$ for general probability statements, 
which may be reduced to $P_\q$ or $Q_\q$ if the event involves only $X$ or $Z$.

\subsection{Frequentist definitions}
The most common frequentist definition is the $(\d,\a)$-tolerance region. For a set
$R$, abbreviate $Q_\q(R)=\Pr_\q(Z\in R)$. Then $\R(X)$ is an \emph{$(\d,\a)$-tolerance region} if
\begin{equation}
\label{EqFreqadTol}
P_\q\Bigl(x: Q_\q\bigl(\R(x)\bigr)\ge 1-\d\Bigr)\ge1-\a,\qquad \forall \q.
\end{equation}
If we let $Q_\q\bigl(\R(X)\bigr)$ denote the probability of $\R(X)$ under $Q_\th$, for $X$ held
fixed, then we can also write the display in the shorter form
$\Pr_\q\bigl(Q_\q\bigl(\R(X)\bigr)\ge 1-\d\bigr)\ge1-\a$, where
the outer probability $\Pr_\q$ refers to $X$, and the inequality must hold for all
possible values of the parameter $\q$. The latter reminds us of the definition of confidence sets,
and indeed it can be seen that \emph{$\R(X)$ is a frequentist \emph{$(\d,\a)$-tolerance region} if and only if
the set $\C(X)=\{\q: Q_\q\bigl(\R(X)\bigr)\ge 1-\d\}$ is a confidence set for $\q$ of confidence level $1-\a$.}

An alternative is the \emph{$\a$-expectation tolerance region}, which requires that
\begin{equation}
\label{EqFreqaTol}
\int Q_\q\bigl(\R(x)\bigr)\,dP_\q(x)\ge 1-\a,\qquad \forall \q.
\end{equation}
With the notational convention as before, the display can be written in the shorter form
$\E_\q Q_\q\bigl(\R(X)\bigr)\ge 1-\a$, which is again required for all possible parameter values.

Both definitions have the form of requiring that 
$\E_\q\ell\bigl[Q_\q\bigl(\R(X)\bigr)\bigr]\ge 1-\a$, for all $\q$, and some
given loss function $\ell$. In the two cases this loss function
is given by $\ell(q)=1\{q\ge 1-\d\}$ for \eqref{EqFreqadTol}, and $\ell(q)=q$ for \eqref{EqFreqaTol}, 
respectively, where $1\{A\}$ is the indicator function of a set $A$.

\subsection{Bayesian definitions}
In the Bayesian setup the parameter $\q$ is generated from a prior distribution $\Pi$,
and the densities $p_\q$ and $q_\q$ are the conditional densities of $X$ and $Z$
given $\q$, respectively. To proceed, it is necessary to make further assumptions that
fix the joint law of $(\q, X, Z)$. The typical assumption is that $X$ and $Z$ are independent given
$\q$.

A natural Bayesian approach is to refer to the predictive distribution of $Z$, and define
a tolerance region $\R(X)$ to be a set such that $\Pr\bigl(Z\in\R(X)\given X\bigr)\ge1-\a$,
i.e. a credible set in the posterior law of $Z$ given $X$. The inequality can be written in terms of the
posterior distribution $\Pi(\cdot\given X)$ of $\q$ given $X$ as
$$\int\Pr\bigl(Z\in \R(X)\given X,\q\bigr)\,d\Pi(\q\given X)\ge1-\a.$$
Under the conditional independence assumption this becomes 
\begin{equation}
\label{EqBayesianaTol}
\int Q_\q\bigl(\R(X)\bigr) \,d\Pi(\q\given X)\ge 1-\a.
\end{equation}
This is like a frequentist $\a$-expectation tolerance region \eqref{EqFreqaTol}, 
but with the expectation with respect
to $X$ under $P_\q$ replaced by the expectation 
with respect to $\q$ under the posterior distribution.

An alternative, derived from a utility analysis by Aitchison \cite{AitchisonBTR}, is the \emph{Bayesian
$(\d,\a)$-tolerance region}, which is a set $\R(X)$ such that
\begin{equation}
\label{EqBayesianadTol}
\Pi\Bigl(\q: Q_\q\bigl(\R(X)\bigr)\ge 1-\d\given X\Bigr)\ge1-\a.
\end{equation}
This may be compared to \eqref{EqFreqadTol}. We can also say that 
\emph{$\R(X)$ is a Bayesian \emph{$(\d,\a)$-tolerance region} if and only if
the set $\C(X)=\{\q: Q_\q\bigl(\R(X)\bigr)\ge 1-\d\}$ is a credible set at level $1-\a$.}

Both types of Bayesian regions satisfy $\int \ell\bigl[Q_\q\bigl(\R(X)\bigr)\bigr]\,d\Pi(\q\given X)\ge 1-\a$, 
for the appropriate loss function $\ell$. Solving the region $\R(X)$ from such an equation
may seem daunting, but good approximations may be easy to obtain using stochastic simulation.
This is true even for complicated data models, as long as one is able to generate a sample from the posterior
distribution given $X$, for instance by implementing an MCMC procedure. 
We make this concrete in Section~\ref{SectionComputation} for a Gaussian variable $Z$, and
illustrate this in Section~\ref{SectionLMM} for an unbalanced linear mixed model (LMM).

\subsection{Comparison}
The frequentist and Bayesian definitions differ in the usual way in 
that the frequentist probabilities in \eqref{EqFreqaTol} and \eqref{EqFreqadTol} refer to the possible values of $x$ in the sample
space, whereas the Bayesian probabilities in \eqref{EqBayesianaTol} and \eqref{EqBayesianadTol} condition on the observed value
of $X$ and refer to the distribution of the parameter. 

As is the case for credible sets versus confidence sets,
the Bayesian approach may feel more natural.

An advantage of the Bayesian approach
is that while the form of the tolerance set $\R(X)$ in \eqref{EqBayesianadTol} is determined by the variable
$Z$, through the prediction problem $Q_\q$, the model for the data $X$ enters only through the
posterior distribution $\Pi(\q\in\cdot\given X)$. If in the former the dependence 
on the parameter $\q$ is not too complicated, then the problem is solvable for even complicated 
data models. In contrast, the frequentist problem permits explicit solutions only in very special
cases, although approximations and asymptotic expansions may extend their use (see \cite{Gaurav2012}). 

Neither the frequentist nor the Bayesian formulation restrains the shape of the region $\R(x)$.
One may prescribe a fixed form and/or seek to optimize the shape with respect to an additional criterion,
such as the volume of the region. The Bayesian formulation is again easier to apply, as the optimization 
will be given the data $X$. In the case of frequentist region it may be necessary to optimize an expected
quantity instead.

In general the two approaches give difference tolerance regions, but the difference may disappear in the large sample limit. 
The requirements of the frequentist and Bayesian \emph{tolerance regions $\R(X)$ for loss function $\ell$ and level $\a$}, 
can be given symmetric formulations, as:
\begin{align}
\label{EqFreqellTolerance}
\E\Bigl(\ell\bigl[Q_\q\bigl(\R(X)\bigr)\bigr]\given \q\Bigr)&\ge 1-\a,\qquad\forall\q,\\
\E\Bigl(\ell\bigl[Q_\q\bigl(\R(X)\bigr)\bigr]\given X\Bigr)&\ge 1-\a.
\label{EqBayesianellTolerance}
\end{align}
In the first the expectation is taken with respect to $X$, which gives an integral over
$x$ with respect to the density $p_\q$. In the second the expectation is relative to
$\q$, which leads to an integral relative to the posterior distribution given $X$. The integral of
the first relative to the prior is identical to the integral of the second
relative to the Bayesian marginal distribution of $X$, but there is no reason that \eqref{EqFreqellTolerance}
implies \eqref{EqBayesianellTolerance} or vice versa. In particular, 
a Bayesian tolerance region need not be a frequentist tolerance region.

However, Bayesian and frequentist inference typically merge if the informativeness
in the data tends to a limit. For instance, this is true for regular parametric models
in the sense that Bayesian  credible sets are frequentist confidence sets, in the limit, with
corresponding levels. The prior is then washed out 
and the Bayesian credible sets are equivalent to confidence sets based on the maximum
likelihood estimator. This equivalence extends to tolerance regions,
under some conditions. We defer a discussion to Section~\ref{SectionFrequentistJustification}.

\subsection{One-sided and two-sided tolerance intervals}
\label{OneTwoSidedInterval}
For a one-dimensional future variable $Z$ it is natural to choose
$\R(x)$ an interval in the real line. The endpoints of such an interval are referred to as \emph{tolerance limits}.

The single finite tolerance limit of a Bayesian \emph{one-sided interval} 
is determined by meeting the $(\d,\a)$- or $\a$-tolerance criterion.
The pair of tolerance limits of a Bayesian \emph{two-sided interval} might be optimized to 
give an interval of minimal length, next to requiring that the tolerance criterion is met.

One-sided tolerance limits possess a straightforward interpretation and implementation.
In particular, the $(\d,\a)$-type has a simple description in terms of confidence intervals
and posterior quantiles:
\begin{compactitem}
\item $(-\infty, U(X)]$ is a frequentist {$(\d,\a)$-tolerance interval} if and only if
it is a $(1-\a)$-{confidence interval} for the induced parameter
$Q_\q^{-1}(1-\d)$; it is a Bayesian {$(\d,\a)$-tolerance interval} if $U(X)$ is the 
$(1-\a)$-quantile of the posterior distribution of  $Q_\q^{-1}(1-\d)$ given $X$.
\item $[L(X),\infty)$ is a frequentist {$(\d,\a)$-tolerance interval}
if and only if it is a $(1-\a)$-confidence interval for $Q_{\q+}^{-1}(\d)$;
it is a Bayesian {$(\d,\a)$-tolerance interval} if $L(X)$ is the $(1-\a)$-quantile of the posterior distribution 
of $Q_{\q+}^{-1}(\d)$.
\end{compactitem}
Here $Q_\q^{-1}(u)=\inf\{z:Q_\q(-\infty,z]\geq u\}$ is the usual quantile function of $Z$,
and $Q_{\q+}^{-1}(u)$ is the right-continuous version of this quantile function. 
(The distinction between the two is usually irrelevant, and linked to the arbitrary convention of including the
boundary point in the tolerance intervals.)
The assertions follow by inverting the inequality $Q_\q\bigl(\R(X)\bigr)\ge 1-\d$, using the fact that
for a cumulative distribution function $Q$ and its quantile function:
$Q^{-1}(u)\le x$ if and only if $u\le Q(x)$, and $Q_+^{-1}(u)< x$ if and only if $u< Q(x-)$,
for $u\in [0,1]$ and $x\in\RR$.

A valid two-sided interval might be constructed as the intersection of two one-sided intervals, each at half of the 
error rate, but this will be conservative and lead to needlessly wide intervals. It makes
good sense to try and minimize the length of the interval. We consider this
in the next section for the case that the future observation $Z$ is univariate Gaussian.

\section{Normally distributed future observation}
\label{SectionUnivariateZ}
Consider the case that the future observation $Z$ is univariate Gaussian with mean
$\n$ and variance $\t^2$. Thus  the parameter is $\q=(\n,\t)$,  and $Z\given X,\n,\t\sim Q_\q=N(\n,\t^2)$.
The probability that the future observation is captured within a candidate tolerance interval $[L,U]$ is
$$Q_\q[L,U]=\Phi\Bigl(\frac{U-\n}\t\Bigr)-\Phi\Bigl(\frac{L-\n}\t\Bigr).$$
It is convenient to parametrize the interval $[L,U]$ by its midpoint $A=(L+U)/2$ and half length $B=(U-L)/2$. For
given $[L,U]$, or $\{A,B\}$, and $\d\in (0,1)$, define the set
$$G_{A,B,\d}=\bigl\{\q=(\n,\t): Q_\q[L,U]\ge 1-\d\bigr\}.$$
For given $[L,U]$ the set $G_{A,B,\d}$ is shaped as in Figure~\ref{FigureOne}.
It  is symmetric about the vertical line $\n=(L+U)/2$ and intersects the horizontal axis $\t=0$ in the interval $[L,U]$.
Changing $A$ moves the set $G_{A,B,\d}$ horizontally, 
while changing $B$ changes its shape, with bigger $B$ making the set both wider and taller. 
Although we use the normal distribution as our example, 
similarly shaped sets and conclusions would be obtainable for other unimodal symmetric distributions.

\begin{figure}
\centerline{\includegraphics[height=4cm]{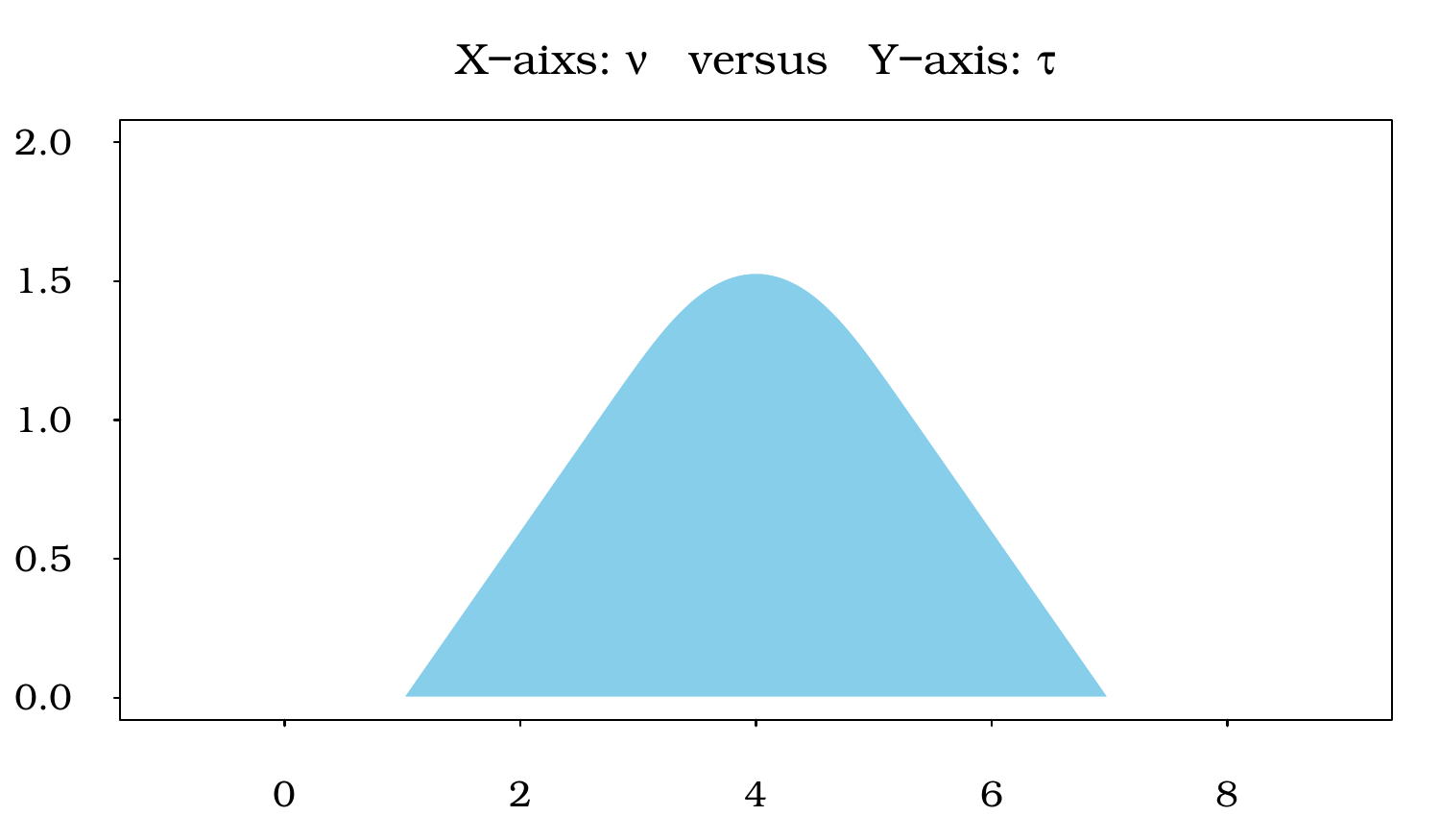}}
\caption{\footnotesize{The set $G_{A,B,\d}$ of pairs $(\n,\t)$ 
such that $\Phi\bigl((U-\n)/\t\bigr)-\Phi\bigl((L-\n)/\t\bigr)\ge 1-\d$, for $A=4$, $B=3$, $\d=0.1$. 
The number $A$ is its horizontal point of symmetry and $B$ is the half-length of its base.
The base of the set (the line segment at height $\t=0$) corresponds to the tolerance interval $[L,U]$.}}
\label{FigureOne}
\end{figure}

It follows that  $\R(X)=\bigl[L(X),U(X)\bigr]$ satisfies
inequality \eqref{EqFreqadTol}  and hence is a frequentist $(\d,\a)$-tolerance interval for $Z$ if and only if
$$P_\q\bigl(x: \q\in G_{A(x), B(x),\d}\bigr)\ge 1-\a,\qquad\forall \q.$$
In other words $[L(X),U(X)]$ is an $(\d,\a)$-tolerance interval for $Z$ if and only if
$G_{A(X),B(X),\d}$ is an $(1-\a)$-confidence region for $\q=(\n,\t)$.
Setting a joint confidence interval for location and dispersion is a familiar problem,
but here the shape is restrained to the form $G_{A,B,\d}$ and the focus will be
on minimizing $B=B(X)$ (in some average sense). Solutions will depend on the type of data $X$.
Standard solutions are available in closed form for the simplest models, and more generally
as approximations. 

Similarly $\R(X)=\bigl[L(X),U(X)\bigr]$ satisfies inequality \eqref{EqBayesianadTol} and hence is a $(\d,\a)$-Bayesian 
tolerance interval for $Z$ if 
\begin{equation}
\label{EqBayesianTolaranceZ}
\Pi\bigl(\q: \q\in G_{A(X),B(X),\d}\given X\bigr)\ge1-\a.
\end{equation}
It is natural to choose $A(X)$ and $B(X)$ to satisfy this inequality in such a way that
$B(X)$ is minimal. In the resulting optimization problem the posterior distribution  $\Pi(\q\in \cdot\given X)$ is a fixed
probability distribution on the upper half plane and optimization entails shifting and scaling the shape shown 
in Figure~\ref{FigureOne} in a position such that it captures posterior mass at least $1-\a$, 
meanwhile minimizing its width. In Section~\ref{SectionComputation} we show how to achieve this numerically 
given a large sample from the posterior distribution.  

The data $X$ determines the posterior distribution, but does not enter the optimization problem.
The parameter $\q$ may not be the full parameter characterising the distribution of $X$, but
our computational strategy will work as long as $\q$ is a function of this full parameter. 
For instance, if $X$ follows a linear regression model with predictor ``time'' 
and $Z$ is an observation at time 0, then $\nu$ will be a function of the regression intercept;
if $X$ follows a random-effects model, then typically $\n$ will depend on the fixed effects
and $\t^2$ will be a specific linear combination of the variance components, depending on practical interests.   

Frequentist methods typically choose $A(X)$ equal to a standard estimator of $\nu$. 
One might guess that the Bayesian solution will be to take
$A(X)$ equal to the posterior mean $\E(\nu\given X)$ of $\nu$.
This would be convenient as it would reduce the optimization of $(A,B)$ to  
the problem of only optimizing $B$.
However, the posterior mean does not necessarily give the minimal length interval.
The following lemma gives a sufficient condition.

\begin{lemma}
\label{LemmaToleranceIntervalAtPosteriorMean}
Suppose that the conditional distribution of $\n$ given $(X,\t)$ is unimodal and symmetric with decreasing density to the 
right of its mode and has mean $\E(\nu\given X,\t)$ that  is free of $\t$. Then
the shortest $(\d,\a)$-Bayesian tolerance interval $[L,U]$ for a future variable
$Z\given X,\n,\t\sim N(\n,\t^2)$ is centered at the posterior mean $\E (\n\given X)$.
\end{lemma}

\begin{proof}
We can decompose the probability on the left side of \eqref{EqBayesianTolaranceZ} as
$$\Pi\bigl(G_{A,B,\d}\given X\bigr)=\int\Pi\bigl(\n\in (G_{A,B,\d})_\t\given X,\t\bigr)\,d\Pi(\t\given X),$$
where $(G_{A,B,\d})_\t=\{\n: (\n,\t)\in G_{A,B,\d}\}$ is the section of $G_{A,B,\d}$ at height $\t$. By the unimodality and
monotonicity the integrand is maximized over $A$ for every $\t$ and a given $B$ by choosing $A=\E (\n\given X,\t)$. 
If this does not depend on $\t$, then this common maximizer $A$
will maximize the whole expression. Since we need to determine $A$ and $B$ so that the expression
is at least $1-\a$, maximizing it over $A$ will give the minimal $B$.
By assumption  $A=\E (\n\given X,\t)=\E (\n\given X)$.
\end{proof}

The condition of the preceding lemma is not unreasonable, but depends on the prior, as illustrated in the
following simple example. In Section~\ref{SectionLMM} we show that for a linear mixed model
the condition is approximately satisfied. Then choosing $A$ equal to the
posterior mean is a fast computational shortcut that may perform almost as well as the optimal solution.

\begin{example}
\label{ExampleNormaliid}
\rm
The simplest possible data model is to let  $X=(X_1,\ldots, X_n)$ be a random
sample from the $N(\n,\t^2)$-distribution. This example was already discussed
by Aitchison \cite{AitchisonBTR}. Here we highlight the implications of Lemma~\ref{LemmaToleranceIntervalAtPosteriorMean}.

For the standard priors $\t^2\sim IG(\a_0,\b_0)$ and $\n\given \t\sim N(a,\t^2/b)$, the conditional posterior distribution of $\n$ given $\t$
is normal with mean
$$\E\bigl(\n\given \t, X\bigr)=\frac{ba+n\bar X}{b+n}.$$ 
Since this is independent of $\t$.
the preceding discussion shows that the shortest $(\d,\a)$-tolerance interval is centred at the posterior mean of $\n$.

Choosing the prior variance $\var(\n\given \t)$ proportional to $\t^2$, which is customary, is crucial for this finding. 
For instance, if we set the prior $\n$ to be independent of $\t$, say $\nu\given \t\sim N(a,b^{-1})$, then
the conditional posterior mean changes to 
$$\E\bigl(\n\given \t, X\bigr)=\frac{ba\t^2+n\bar X}{b\t^2+n}.$$ 
This is $\bar X$ if $\t=0$ and shrinks to the prior mean $a$ as $\t\ra\infty$.
For illustration, let $a=0,\: b=0.1,\: \a_0=\b_0=0.01$.
Given data with  $n=3,\: \bar{x}=10,\: s^2=1$, we approximated 
the posterior distribution of $(\n,\t)$ given $X$ by a Gibbs sampler,
using the full conditional posteriors, where $s^2=\sum{(x_i-\bar{x})^2}/(n-1)$:
\begin{align*}
\n\given \t, X&\sim N\Big(\frac{ba\t^2+n\bar X}{b\t^2+n},\frac{\t^2}{b\t^2+n}\Bigr),\\
\t^2\given \n, X&\sim IG\Big(\a_0+\frac{n}{2},\b_0+\frac{(n-1)s^2+n(\bar X-\n)^2}{2}\Big).
\end{align*}
The contour plots of the posterior distribution in the left panel of Figure~\ref{Figure2} show 
dependence between $\n$ and $\t$ given $X$, and ensuing functional dependence of $\E(\n\given \t,X)$ on $\t$. 
Using Algorithm~\ref{AlgoProposed}, as explained in Section~\ref{SectionComputation}, 
we computed the shortest tolerance interval (i.e.\ the smallest $\hat{B}$) for every possible location $A$
of the interval, for $A$ in a neighbourhood of the posterior mean $\E(\n\given X)$. 
The right panel of Figure~\ref{Figure2} shows $\hat B$ as a function of $\hat A$. The minimum value
is \emph{not} taken at  the location of the posterior mean $\E(\nu\given X)$, which is
indicated by a dashed line. 

Admittedly the data in this example has been tweaked to illustrate the principle.
Inspection of the vertical scale for $\hat B$ shows that the global minimal length 
of a tolerance interval is only slightly smaller than the length of the minimal interval centred at the posterior mean of $\nu$.
In  Section~\ref{SectionLMM} we theoretically investigate a similar phenomenon for linear mixed models,
and in Section~\ref{SectionFrequentistJustification} we study this approximation in a large sample context.
\end{example}

\begin{figure}
\centerline{\includegraphics[height=4cm]{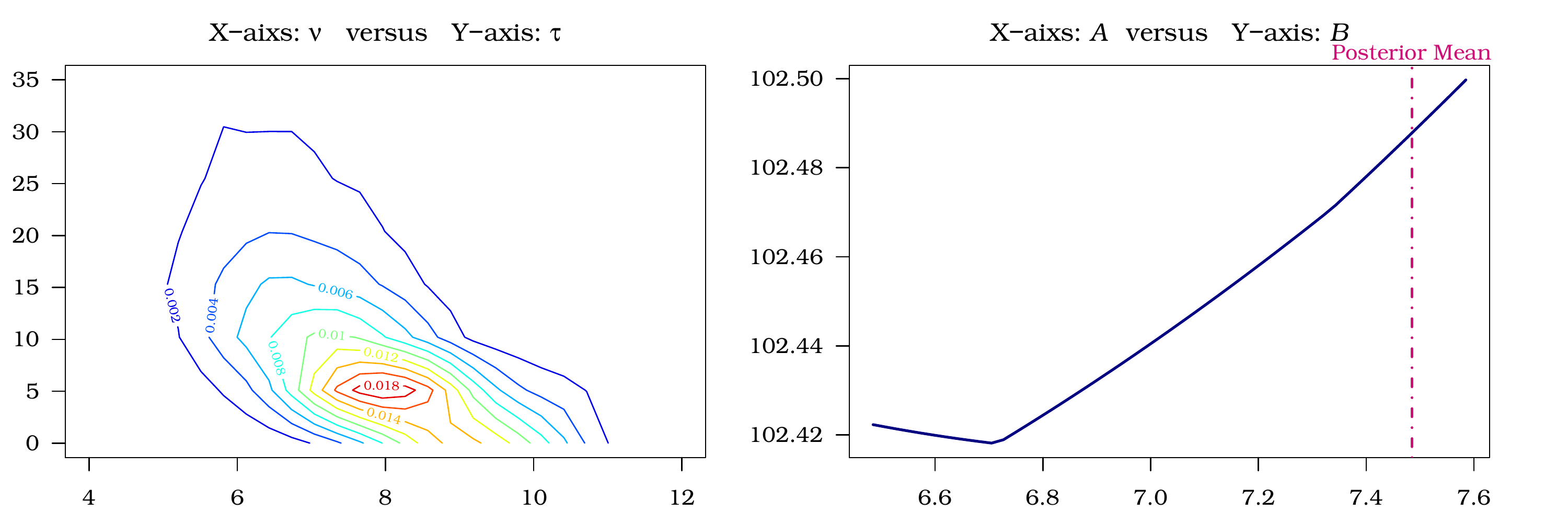}}
\caption{\footnotesize{
Left panel: density-level contour plots of MCMC approximation to realization of posterior distribution $\n,\t\given X$ in Example~\ref{ExampleNormaliid} 
(prior parameters: $a=0,\: b=0.1,\: \a_0=\b_0=0.01$; data $\bar{x}=10$, $s^2=1$, $n=3$).
Right panel: corresponding half-lengths $B$ (vertical axis) of the $(\d=0.05, \a=0.1)$-tolerance interval centered at $A$
(horizontal axis); the minimal length is not taken at the posterior mean  $\E(\n\given X)$, whose location is
indicated by the abscissa of the dotted line.}}
\label{Figure2}
\end{figure}

\subsection{Computational strategy}
\label{SectionComputation}
In this section we elaborate on the computation of the two-sided Bayesian $(\d,\a)$-tolerance interval 
for a normally distributed univariate future variable $Z$, as discussed in Section~\ref{SectionUnivariateZ}.
We also compare our approach to the ones taken in \cite{wolfinger} and \cite{krishnamoorthy2009}.
We assume given a large sample of values $\q_j=(\n_j,\t_j)$ from the posterior distribution of $\q=(\n,\t)$ given $X$.
This could be the result of an MCMC run of a sampler for the posterior distribution, or, depending on the data model,
of using an analytic formula for the posterior distribution. We shall use the sample values $(\n_j,\t_j)$ to approximate expectations
under the posterior distribution, whence they need not be independent, and values from a (burnt-in) MCMC run will qualify.
Possible dependence together with sample size will determine the error due to simulation.

The idea is to replace the posterior distribution in \eqref{EqBayesianadTol} or \eqref{EqBayesianTolaranceZ}
by the empirical distribution of the values $\{\q_j\}_{j=1}^J$. For
a given interval $\R(X)=[L,U]$ we can (in theory) compute the values $Q_{\q_j}[L,U]$ and next search for the
interval $[L,U]$ of minimal length $U-L$ such that 
$$\frac 1J\# \bigl\{1\le j\le J:  Q_{\q_j}[L,U]\ge 1-\d\bigr\}\doteq 1-\a,$$
where $\doteq$ means approximately equal, yielding a slightly conservative or anti-conservative
solution in case exact equality is not attainable due to discretization. 

Wolfinger \cite{wolfinger} proposed the algorithm, summarized as Algorithm~\ref{AlgoWKM} in the display, and this
was refined (or corrected) by Krishnamoorthy and Mathew \cite{krishnamoorthy2009}, Chapter~11. 
This algorithm has a convenient graphical representation and 
has been widely adopted in practice.  The idea is to compute for every $\q_j$
the quantiles $L_j=Q_{\q_j}^{-1}(\d/2)$ and $U_j=Q_{\q_j}^{-1}(1-\d/2)$, yielding intervals $[L_j,U_j]$
with $Q_{\q_j}[L_j,U_j]\ge 1-\d$, and next setting the tolerance interval $[L,U]$ equal to 
an interval that is symmetric about the posterior mean and contains a fraction $1-\a$ 
of the intervals $[L_j,U_j]$ (Krishnamoorthy and Mathew), or is contained in a fraction $\a$ of these intervals 
(Wolfinger).  
The graphical interpretation is to plot the points $(U_j,L_j)$ in the $x$-$y$-plane and search for a point
$(U,L)$ on the line $y+x=2\hat\n$, for $\hat\n$ the posterior mean or some other useful estimator,
such that a fraction $1-\a$ of the points are in the left-upper quadrant relative to the point $[L,U]$
(see Figure~11.1 in \cite{krishnamoorthy2009} for an example).
This method results in an interval that is more confident than the prescribed level $1-\a$,
and appears not to optimize the length of the interval.

\begin{figure}
\resizebox{.9\linewidth}{!}{
\begin{minipage}{\linewidth}
\begin{algorithm}[H]
\caption{WKM solution for two-sided tolerance interval}\label{AlgoWKM}
\KwData{Given $\a,\d,\{(\n_j,\t_j)\}_{j=1}^J$}
Let $\hat{A}=\sum_j \n_j/J$ \;
Calculate two quantiles sequences: $\{L_j\equiv Q^{-1}_{\n_j,\t_j}(\frac{\d}{2})\}^J_{j=1}$ 
and $\{U_j\equiv Q^{-1}_{\n_j,\t_j}(1-\frac{\d}{2})\}^J_{j=1}$  \;  
Find a point $(\hat L,\hat U)$  such that $\hat L+\hat U=2\hat A$\ satisfying one of the following \;
   $\quad$(W) $\arg\min_{\hat L,\hat U} \left|\frac{\#S}{J}-\a \right|$,
                       where $S=\{(L_j,U_j):L_j\leq\hat L,  U_j\geq\hat U\}$\;
   $\quad$(KM) $\arg\min_{\hat L,\hat U} \left|\frac{\#S}{J}-1+\a \right|$,
                       where $S=\{(L_j,U_j):L_j\geq\hat L, U_j\leq\hat U\}$\;
\KwResult{two-sided tolerance interval $[\hat L,\hat U]$}  
\end{algorithm}
\end{minipage}
}
\end{figure}

Here we propose another algorithm that directly utilizes \eqref{EqBayesianadTol}.
We seek to minimize $B$ under the constraint that the interval
$[L,U]=[A-B,A+B]$ satisfies \eqref{EqBayesianadTol}. This takes two steps: 
for fixed $A$ we  optimize over $B$; next we perform a grid search over $A$.
Because given $A$, the optimizer over $B$ will yield equality in \eqref{EqBayesianadTol},
$\hat B$ will be the solution to 
\begin{equation}
\Pi\biggl[    \Phi\Bigl(\frac{{A}+B-\n}{\t}\Bigr)
              -\Phi\Bigl(\frac{{A}-B-\n}{\t}\Bigr)\geq 1-\d  \given X \biggr]=1-\a.
\label{EqTarget}
\end{equation}
The posterior mean $\E(\n\given X)$ will typically be close to the optimal solution for $A$,
and is a good starting point for this parameter. As a fast approximation we may also set $A$
equal to this value and omit the grid search.

In practice, we replace the posterior distribution in equation \eqref{EqTarget} by an average over the sample
values $(\n_j,\t_j)$. The posterior mean of $\n$ can be approximated by the average of the sample values $\n_j$.
Given $\hat A$ we approximate $\hat B$ by the 
$(1-\a)$-quantile of the points $g_j$ computed as the solutions to
\begin{equation}
\label{FunctionG}
Q_{\n_j,\t_j}[\hat A-g_j,\hat A+g_j]
\equiv\Phi\Bigl(\frac{\hat A+g_j-\n_j}{\t_j}\Bigr)-\Phi\Bigl(\frac{\hat A-g_j-\n_j}{\t_j}\Bigr)=1-\d.
\end{equation}
The motivation for this procedure is that $(\n_j,\t_j)\in G_{\hat A,\hat B,\d}$ if and only if $Q_{\n_j,\t_j}[\hat A-\hat B, \hat A+\hat B]\ge 1-\d$,
whence precisely the points $(\n_j,\t_j)$ with $g_j\le \hat B$ satisfy 
$Q_{\n_j,\t_j}[\hat A-\hat B,\hat A+\hat B]\ge 1-\d$ and hence are inside the set 
$G_{\hat A,\hat B,\d}$, whereas the other points are outside this set. 
This makes the posterior mass of the set equal to $1-\a$ up to simulation error.

\begin{figure}
\centerline{\includegraphics[height=4cm]{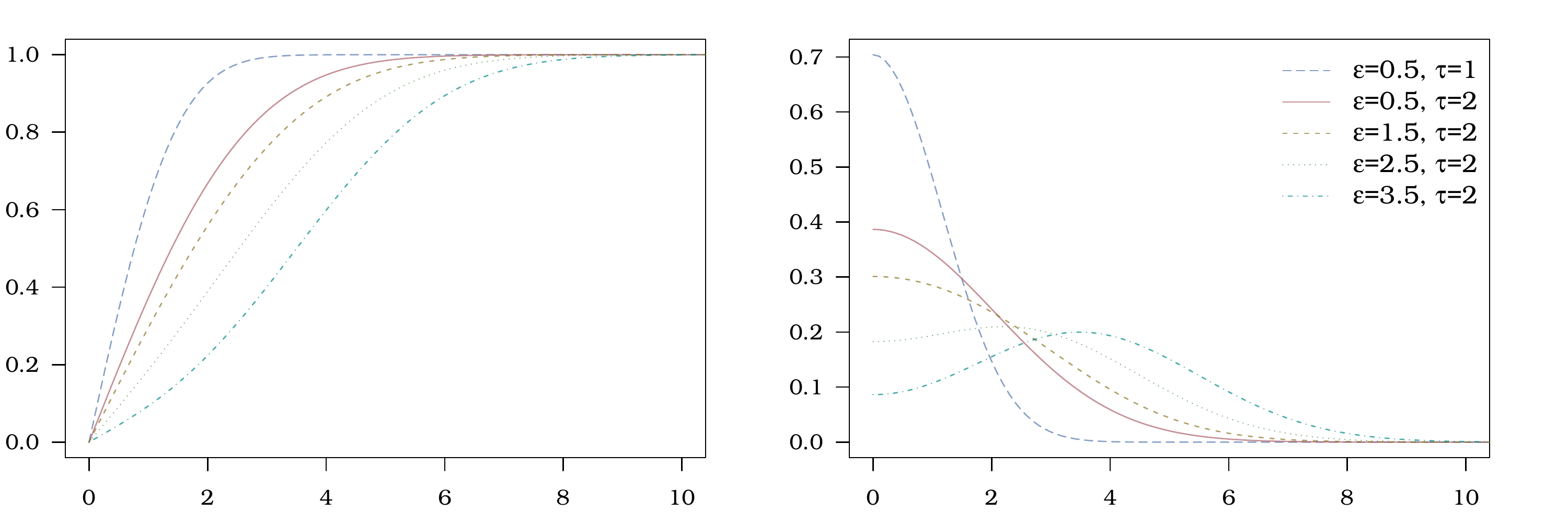}}
\caption{\footnotesize{
Left panel: plot of the curves
$g\mapsto \Phi(\frac{\e+g}{\t})-\Phi(\frac{\e-g}{\t})$ for various settings of $\e$ and $\t$.
Right panel: derivatives of these curves.}}
\label{Figure3}
\end{figure}

Given $\hat A$ and $(\n_j,\t_j)$, the function $g\mapsto Q_{\n_j,\t_j}[\hat A-g,\hat A+g]$ in \eqref{FunctionG} 
is  increasing, from the value 0 when $g=0$ to 1 as $g\rightarrow\infty$ (see Figure~\ref{Figure3}). 
The solutions $g_j$ to each equation \eqref{FunctionG} can be found fast by a Newton-Raphson algorithm, 
with some caution on choosing the initial value for $g_j$ (the algorithm 
will diverge if the initial value is chosen in the domain where $Q_{\n_j,\t_j}[\hat A-g_j,\hat A+g_j]$ is very close to 1).
An appropriate algorithm is listed in Algorithm~\ref{AlgoProposed}. Note that
the (middle) expression in \eqref{FunctionG} does not change if $\e:=\hat A-\n_j$ is replaced by $-\e$.

\begin{figure}
\resizebox{.9\linewidth}{!}{
\begin{minipage}{\linewidth}
\begin{algorithm}[H]
\caption{Proposed solution for two-sided tolerance interval}\label{AlgoProposed}
\KwData{Given $\a,\d,\{(\n_j,\t_j)\}_{j=1}^J$}
  
  Let $\hat{A}=\sum_j \n_j/J$ \;
  \For{$j=1,2,...,J$}{
	Solve equation \eqref{FunctionG} by a Newton-Raphson algorithm as follows\;  
	\eIf{$|\hat A-\n_j|<\t_j$}{$g_0=|\hat A-\n_j|+\t_j$ 
	}{$g_0=|\hat A-\n_j|$}
	set initial value for $g_j$ at $g_0$\;
	\While{$\o>0.0001$}{
           let $q_{\n_j,\t_j}[\hat A-g_j,\hat A+g_j]$ be the first-order derivative 
           of $Q_{\n_j,\t_j}[\hat A-g_j,\hat A+g_j]$\;
	   $g_j=g_j-\frac{Q_{\n_j,\t_j}[\hat A-g_j,\hat A+g_j]-1+\d}{q_{\n_j,\t_j}
                 [\hat A-g_j,\hat A+g_j]}$\;
	   $\o=Q_{\n_j,\t_j}[\hat A-g_j,\hat A+g_j]-1+\d$;
        }
   }
The above loop results in $\{g_j\}_{j=1}^J$, 
and let $\hat B$ be its $(1-\a)\emph{th}$ sample quantile\;
\KwResult{two-sided tolerance interval $[\hat A-\hat B,\hat A+\hat B]$}  
\end{algorithm}
\end{minipage}
}
\end{figure}

\section{Data from a linear mixed model}
\label{SectionLMM}
In this section we apply the preceding to a model that is representative for  practice in pharmaceutical quality control:
the \emph{linear mixed model} (LMM). We assume that the data $X$ are acquired in an LMM design,
and that the future variable $Z$ is defined in terms of the same LMM.
We concentrate attention to the two-sided $(\d,\a)$-tolerance interval.

In the LMM we observe a vector $X=U\b+V\g+e$, for known (deterministic) matrices
$U$ and $V$ of covariates, a vector of fixed effects parameters $\b$, an unobserved random effect vector $\g$,
and an error vector $e$. Assume that $\g$ and $e$ are independent, with
\begin{equation}
\label{EqLMM}
\g\sim N(0,D),\qquad e\sim N(0,\s^2 I).
\end{equation}
Then the data  $X$ follows a $N(U\b, C)$-distribution, for $C=VDV^T+\s^2I$,
and the full parameter is $(\b,D,\s^2)$.

Consider predicting a new observation $Z=u^T\b+v^T\g'+e'$ with given fixed and random effects coefficients $u$ and 
$v$ and \emph{newly generated} random effect vector $\g'$ and error $e'$, with $\g'\sim\g$ and $e'\sim N(0,\s^2)$. 
Thus $\g'$ is assumed equal in distribution to, but independent of $\g$, and
similarly for $e'$. This target for prediction is reasonable in many contexts, but sometimes another choice, in particular for $\g'$,
may be more relevant. Typically $\g$ will carry a group structure matched by a block structure in $V$. The vector $v$ will
then have nonzero coordinates corresponding to a single group.
The target corresponds to setting the distribution $Q_\q$ of $Z$ equal to $N(\n,\t^2)$, with
\begin{equation}
\label{EqNuTau}
\n=u^T\b, \qquad \t^2=v^TDv+\s^2.
\end{equation}
The ``prediction'' parameter $\q=(\n,\t)$ is of smaller dimension than
the full parameter'' $(\b,D,\s^2)$ governing the distribution of the data $X$, 
whence part of the latter full parameter can be  considered a nuisance parameter.
To set a Bayesian tolerance interval we need a posterior distribution of $\q$ given the data $X$. This will typically
be inferred from a posterior distribution of the full parameter, resulting from a prior distribution on $(\b,D,\s^2)$.

The (conditional) posterior distribution for a conditional prior $\b\given D,\s^2\sim N(0,\Lambda)$, 
where $\Lambda$ may depend on $(D,\s^2)$, satisfies
$$\b\given D,\s^2,X\sim
N\Bigl((U^TC^{-1}U+\Lambda^{-1})^{-1}U^TC^{-1}X, (U^TC^{-1}U+\Lambda^{-1})^{-1}\Bigr).$$
(An alternative expression for the posterior mean is $\Lambda U^T(U\Lambda U^T+C)^{-1]}X$.)
In general $\E(\nu\given D,\s^2,X)=u^T\E(\b\given D,\s^2, X)$  will depend on $D$ and $\s^2$ (hidden in $C$) and hence
typically also on $\t^2$, in view of \eqref{EqNuTau}. 
Therefore Lemma~\ref{LemmaToleranceIntervalAtPosteriorMean} does not apply, 
and there appears to be no reason that a shortest tolerance interval 
would be centered at the posterior mean of $\nu$. To obtain the shortest interval,
Algorithm~\ref{AlgoProposed} should be augmented with a search on possible centerings $A$.

As in the standard \emph{i.i.d.} model in Example~\ref{ExampleNormaliid},
the dependence on $\s^2$ can be removed by choosing the variances $\Lambda$ and $D$ proportional to
$\s^2$. If $D=\s^2D_0$ and $\Lambda=\s^2\Lambda_0$, 
then $C$ will be $\s^2(VD_0V^T+I)=:\s^2 C_0$ and the conditional posterior mean of $\b$
will be $(U^TC_0^{-1}U+\Lambda_0^{-1})^{-1}U^TC_0^{-1}X$.
However, the dependence on $D_0$ (through $C_0$) remains, in general.

Letting the prior covariance matrix $\Lambda$ tend to infinity corresponds to a non-informative prior on $\b$.
If all other quantities are fixed and $\Lambda\ra\infty$, then 
$$\E(\b\given D,\s^2,X)\ra (U^TC^{-1}U)^{-1}U^TC^{-1}X.$$
The limit is the maximum likelihood estimator of $\b$ in the model where $C$ is known.
Since this is still dependent on $C$ (and hence $D$ and $\s^2$), it seems 
that for both the Bayesian and frequentist tolerance intervals the two parameters $\n$ and $\t$ 
cannot be separated in general. The choice $\Lambda=\l (U^TC^{-1} U)^{-1}$ leads to $\l/(1+\l)$
times the maximum likelihood estimator, and hence also still depends on $C$.

\subsection{Approximations to the conditional posterior mean}
For special designs the dependence of $\E(\b\given D,\s^2,X)$ on $(D,\s^2)$ is only mild and can be quantified. 
We discuss three examples of linear mixed models. For clarity we restrict to balanced designs, although
its flexibility makes the Bayesian approach particularly valuable for unbalanced cases, as is illustrated by
the  numerical examples in Section~\ref{SectionNumericalExamples}.

\begin{example}
[One-way random effects]
\rm
Suppose $X$ is a vector with coordinates $X_{ik}=\b+\g_i+e_{ik}$, for $i=1,\ldots,m$ and $k=1,\ldots,n$,
ordered as $(X_{11},\ldots, X_{1n}, X_{21},\ldots, X_{mn})$, where $\b\in \RR$ and 
$\g=(\g_1,\ldots, \g_m)^T$ with i.i.d.\ $\g_i\sim N(0,d^2)$, so that $D=d^2 I_m$,
for $I_m$ the $(m\times m)$-identity matrix.  As prior on $\b$ we choose a one-dimensional normal distribution $N(0,\l^2)$.

The matrix $U$ is the $mn$-vector $1_{mn}$ with
all coordinates equal to 1, while $V$ is the $(mn\times m)$-matrix with $i^{th}$  column having 1s in rows
$(i-1)n+1$ to $in$ and 0s in the other rows. Then $V^TV=n I_m$,
and $U^TV=n1_m$, and it can be verified that $C1_{mn}=(nd^2+\s^2)1_{mn}$ and hence $C^{-1}U=(nd^2+\s^2)^{-1}1_{mn}$.
The coefficient vector of $\E(\b\given D,\s^2,X)$ is
$$(U^TC^{-1}U+\l^{-2})^{-1}U^TC=\Bigl(mn+\frac{nd^2+\s^2}{\l^2}\Bigr)^{-1}1_{mn}^T.$$
For $\l=\infty$, this is free of $d^2$ and $\s^2$, while for finite, fixed  $\l$ and $m,n\ra\infty$, the
coefficient vector is $(mn)^{-1}\bigl(1+O(d^2/(m\l^2))+O(\s^2/(mn\l^2))\bigr)$. 

The dependence on $d$ and $\s$ can be removed by choosing $d=d_0\s$ and $\l=\l_0\sqrt{nd_0^2+1}\, \s$.
\end{example}

\begin{example}
[Full random effects]
\rm
Suppose $X_{ik}=u_{ik}^T\b+v_{ik}^T\g_i+e_{ik}$, ordered as in the preceding example, but
now with observed covariates $u_{ik}\in\RR^p$ and $v_{ik}\in\RR^q$, fixed effects parameter
$\b\in \RR^p$ and i.i.d.\ random effects $\g_i\sim N_q(0,D_q)$, for $i=1,\ldots, m$. The corresponding matrices
$U$ and $V$ are
$$U=\left(\begin{matrix} U_1\\\vdots\\ U_m\\\end{matrix}\right),\quad
V=\left(\begin{matrix} V_1&\cdots&0\\ \vdots&\ddots&\vdots\\ 0&\cdots&V_m\end{matrix}\right),\quad
U_i=\left(\begin{matrix} u_{i1}^T\\ \vdots\\ u_{in}^T\end{matrix}\right),\quad
V_i=\left(\begin{matrix} v_{i1}^T\\ \vdots\\ v_{in}^T\end{matrix}\right).$$
Then $C=VDV^T+\s^2I$ is an $(mn\times mn)$-block-diagonal matrix with $m$ blocks
$V_iD_qV_i^T+\s^2I_n$, and
\begin{align*}
U^TC^{-1} U&=\sum_{i=1}^m U_i^T(V_iD_qV_i^T+\s^2I_n)^{-1}U_i,\\
U^T C^{-1}&=\Bigl(U_1^T(V_1D_qV_1^T+\s^2I_n)^{-1},\ldots,U_m^T(V_mD_qV_m^T+\s^2I_n)^{-1}\Bigr).
\end{align*}
The matrices $U_i^TV_i$ and $V_i^TV_i$ are of dimensions $p\times q$ and $q\times q$, and are
sums over the $n$ observations (for $k=1,\ldots,n$)  per group $i$, as defined by the random effect $\g_i$. For
conventional asymptotics we could view them as $n$ times a matrix of fixed order. Then, for $D_0=\s^{-2}D_q$,
\begin{align}
\label{EqInverseSeries}
&(V_iD_0V_i^T+I)^{-1}=I-V_i(D_0^{-1}+V_i^TV_i)^{-1}V_i^T\nonumber\\
&\qquad=I-V_i(V_i^TV_i)^{-1}\bigl[D_0^{-1}(V_i^TV_i)^{-1}+I\big]^{-1}V_i^T\nonumber\\
&\qquad=I-V_i(V_i^TV_i)^{-1}\bigl[I-D_0^{-1}(V_i^TV_i)^{-1}+(D_0^{-1}(V_i^TV_i)^{-1})^2+\cdots\bigr]V_i^T\nonumber\\
&\qquad=P_{V_i^\perp}+ V_i(V_i^TV_i)^{-1}D_0^{-1}(V_i^TV_i)^{-1}V_i^T\nonumber\\
&\qquad\qquad\qquad\qquad\quad-\sum_{k=2}^\infty(-1)^kV_i(V_i^TV_i)^{-1}(D_0^{-1}(V_i^TV_i)^{-1})^kV_i^T,
\end{align}
where $P_{V^\perp}$ is the projection on the orthocomplement of the linear span of the columns of $V$.
If $V_i^TV_i$ is large, then the series on the right can be neglected, as its terms contain multiple terms $(V_i^TV_i)^{-1}$.
There is then still dependence on $D$ and $\s^2$ in the second term, which may dominate the first term.

In a full random effects model,
every random effect is matched by a fixed effect with the same covariate vector (supplying a common mean value to
the random effects), and hence $u_{ik}=v_{ik}$, for every $(i,k)$, and consequently $U_i=V_i$. Then $U_i^TP_{V_i^\perp}=0$ 
and $U_i^TV_i(V_i^TV_i)^{-1}=I_p$. If $n^{-1}V_i^TV_i$ stabilizes as $n\ra\infty$,
\begin{align*}
\s^2U^TC^{-1} U&=\sum_{i=1}^m\sum_{k=1}^\infty(-1)^{k-1}D_0^{-1}((V_i^TV_i)^{-1}D_0^{-1})^{k-1}=m\Bigl(D_0^{-1}+O\bigl(\frac 1n\bigr)\Bigr),\\
\s^2U^TC^{-1}&=\sum_{k=1}^\infty(-1)^{k-1} \Bigl(D_0^{-1}(V_i^TV_i)^{-1})^{k}V_i^T\Bigr)_{i=1}^m\\
&=\Bigl(D_0^{-1}(V_i^TV_i)^{-1}\Bigl(I+O\bigl(\frac 1n\bigr)\Bigr)V_i^T\Bigr)_{i=1}^m.
\end{align*}
Thus we find that the coefficient vector of the posterior mean of $\b$ satisfies
\begin{align*}&(U^TC^{-1}U+\Lambda^{-1})^{-1}U^TC^{-1}\\
&\qquad=\biggl(\frac1m\Bigl(I+O\Bigl(\frac{D_0}n\Bigr)+\frac{D_0\Lambda^{-1}}m\Bigr)^{-1}\Bigl((V_i^TV_i)^{-1}\Bigl(I+O(\frac1n\bigr)\Bigr)V_i^T\Bigr)\biggr)_{i=1}^m.
\end{align*}
To first order this is free of $D$ and $\s^2$ with relative remainders of the order $D_0/n$ and $D_0\Lambda^{-1}/m$.
\end{example}

\begin{example}
[Random effects with additional fixed effects]
\rm
In the preceding example every random effect may be matched by a fixed effect with the same covariate vector 
(supplying a common mean value to the random effects), but there more fixed than random effects.
This corresponds to setting $u_{ik}^T=(v_{ik}^T, \bar u_{ik}^T)$, 
which implies $U_i=(V_i,\bar U_i)$, for an $(n\times (p-q))$-matrix $\bar U_i$. 
The formulas in the preceding example must then be adapted to, where the approximations
refer to ignoring the series in \eqref{EqInverseSeries},
for $W_i=V_i(V_i^TV_i)^{-1}$, 
\begin{align*}
\s^2U^TC^{-1} U&\doteq\sum_{i=1}^m \biggl[\left(\begin{matrix} 0&0\\ 0&\bar U_i^TP_{V_i^\perp}\bar U_i\end{matrix}\right)
+ \left(\begin{matrix} D_0^{-1}& D_0^{-1}W_i^T\bar U_i\\
\bar U_i^TW_iD_0^{-1}&\bar U_i^TW_iD_0^{-1}W_i^T\bar U_i\end{matrix}\right)\biggr],\\
\s^2U^TC^{-1}&\doteq\biggl(\left(\begin{matrix} 0\\\bar U_i^TP_{V_i^\perp}\end{matrix}\right)
+\left(\begin{matrix} D_0^{-1}W_i^T\\ \bar U_i^TW_iD_0^{-1}W_i^T\end{matrix}\right)\biggr) _{i=1}^m.
\end{align*}
It is reasonable to expect that the  matrices $\bar U_i^TW_i=\bar U_i^TV_i (V_i^TV_i)^{-1}$ will settle down as, as will
the matrices $n^{-1}\bar U_i^TP_{V_i^\perp}\bar U_i$.  Then up to lower order terms
$$(U^TC^{-1}U)^{-1}U^TC^{-1}\doteq
\sum_{i=1}^m \left(\begin{matrix} D_0^{-1}& D_0^{-1}W_i^T\bar U_i\\
\bar U_i^TW_iD_0^{-1}&\bar U_i^TP_{V_i^\perp}\bar U_i\end{matrix}\right)^{-1}
\biggl(\left(\begin{matrix} D_0^{-1}W_i^T\\ U_i^TP_{V_i^\perp}\end{matrix}\right)\biggr) _{i=1}^m.$$
Here the three appearances of $D_0^{-1}$ in the top row cancel each other (as follows by factorizing out the
block diagonal matrix with blocks $D_0^{-1}$ and $I_{p-q}$ from  the inverse matrix), while the factor
$\bar U_i^TW_iD_0^{-1}$ in the bottom row is a factor $1/n$ smaller in order than the matrix 
$\bar U_i^TP_{V_i^\perp}\bar U_i$ and hence can be set to 0 up to order $1/n$. It follows that again the matrix
is free of $D$ and $\s^2$ up to order $1/n$.
\end{example}

\subsection{Numerical examples}
\label{SectionNumericalExamples}
We evaluate the small sample performance of our proposed algorithm via a simulation study,
with data generated from a one-way random effects model. We observe $(X_{ik}: i=1,2,\ldots, m; k=1,2,\ldots,n_i)$,
where $X_{ik}$ is the $k^{th}$ value of the $i^{th}$ group and satisfies $X_{ik}=\n+\g_i+e_{ik}$
for i.i.d.\ $\g_i\sim N(0,d^2)$ independent of the i.i.d.\ error $e_{ik}\sim N(0,\s^2)$.
We are interested in the two-sided $(\d,\a)$-tolerance interval related to a future observation 
$Z=\n+\g+e\sim Q_\q=N(\n,\t^2)$, where $\g$ and $e$ are independent copies of the $\g_i$ and $e_{ik}$.
The parameter is $\q=(\n,\t^2)$, for $\t^2=d^2+\s^2$.

We used $6$ different parameter settings. In every setting the overall mean was set equal to  $\n=0$, and the group variance  to $d^2=1$.
The intra-correlation $\s^2/(d^2+\s^2)$ was chosen equal to the numbers $0.1,0.3,0.5,0.7,0.9$.
In every setting the number of groups was $m=6$ and the group sizes were $(n_1,\ldots, n_6)=(2,3,4,2,3,4)$.
This simulation setup is the same as in \cite{Gaurav2012}, and facilitates a cross-comparison
against the performance of a standard frequentist solution. 

We computed the $(\d=0.1,\a=0.05)$-tolerance interval $[L,U]$ by Algorithm~\ref{AlgoProposed}, 
for every of $K=1000$ replicates of the data in each parameter setting, 
and computed the true coverage $Q_\q[L,U]$ of these intervals using the true parameter $\q$ of the simulation.
We consider an interval with true coverage no less than $1-\d$ as ``qualified'', and compared
the empirical fraction of qualified intervals out of the $K$ replicates to the nominal value $1-\a$. 
The procedure is considered to perform well in the frequentist sense if this empirical fraction is close
to this nominal value. Here we must allow for the simulation error, which has a standard error of 
$\sqrt{p(1-p)/K}$, for $p$ the true coverage, which is unknown, but hopefully close to $1-\a$.  

For each simulated dataset the posterior distribution of $(\n,d,\s^2)$ was approximated by a standard Gibbs sampler
(with the vector of random effects $\g_i$ added in as a fourth parameter), before utilizing Algorithm~\ref{AlgoProposed}.  
Two setups of priors were deployed, both with independent priors on the three parameters $\n,d,\s$.
The first is the vanilla setup with vague marginal prior distributions 
$\n\sim N(0,1000)$, $d^2\sim IG(0.001,0.001)$ and ${\color{red} \s^2}\sim IG(0.001,0.001)$.
The second uses the same prior on $\s^2$, but uses a $t$-distribution for $\n$ given by the hierarchy
$\n\given \s_0\sim N(0,\s^2_0)$ and $\s^2_0\sim IG(0.001,0.001)$, and the prior on $d$ given by
the structural equation $d=|\xi|\o$ for independent $\xi\sim N(0,1)$ and $\o^2\sim IG(0.001,0.001)$.
The latter specification can also be understood as over-parameterizing the 
distribution of the random effect using two parameters instead of one, as $\g\given \xi,\o\sim N(0,\xi^2\o^2)$.
This ``parameter expansion'' is meant to 
enhance the mixing rate of the Gibbs sampler, in particular when the number of groups $m$ is small.
See \cite{browne2006} for a  comparison of methods (including nonBayesian methods) to fit the LMM.

In all simulation settings the tolerance intervals were constructed both by fixing the center point $A$ 
at the posterior mean $E(\n\given X)$, and by seeking an optimal value of $A$ to minimize the half length $B$ of the interval.
Thus four tolerance intervals were calculated based on each simulated dataset.
To save on computation time the optimization over $A$ was carried out only approximately.
Still for 85\% of the simulation cases a shorter interval was obtained than the interval at the posterior mean,
in a few cases as much as 20\% shorter, but in 75\% of the cases no more than a few percentage points.
Table~\ref{TabBOpt} reports the quotients of the lengths.

The proportions of intervals that attain true coverage at least $1-\d=0.9$, are listed in Table~\ref{TabRes}.
They are reasonably close to the nominal value $1-\a=0.95$, with deviations in both directions up to
several percentage points. The performance seems to depend on the true intra-correlation. This dependence 
follows a similar pattern as for the high-order asymptotic solution in \cite{Gaurav2012} without correction for imbalance, 
and is close to their solution that includes correction when the intra-correlation is small or very big.
The tolerance intervals centered at the (approximately) optimal value have shorter length and
attain lower confidence, but their performance seems to surpass slightly the intervals centered 
at the posterior mean $\E(\n\given X)$. 

The difference in  performance of Algorithm~\ref{AlgoProposed} between the two prior
setups is within the order of the simulation  error.

\begin{table}[ht]
\caption{\footnotesize{Interval length at optimal value of $A$ relative to fixing $A$ at $\E(\n\given X)$}}
\centering
\raff{1.3}
\resizebox{\columnwidth}{!}{%
\begin{tabular}{cccccccccc} \toprule
  & \multicolumn{4}{c}{\emph{under Vanilla setup}} & \phantom{abc} 
  & \multicolumn{4}{c}{\emph{under Parameter Expansion setup}} \\
  \cmidrule{2-5} \cmidrule{7-10}
  & $Min$ & $0.25^{th} qu.$ & $Median$ & $0.75^{th} qu.$ &
  & $Min$ & $0.25^{th} qu.$ & $Median$ & $0.75^{th} qu.$ \\
  \midrule
  $intra-correlation$ \\
  $0.1$  & 0.8420 & 0.9958 & 0.9984 & 0.9993 && 0.8625 & 0.9969 & 0.9991 & 0.9997 \\ 
  $0.3$  & 0.8244 & 0.9943 & 0.9984 & 0.9994 && 0.8569 & 0.9969 & 0.9989 & 0.9997 \\ 
  $0.5$  & 0.8192 & 0.9905 & 0.9981 & 0.9992 && 0.8374 & 0.9954 & 0.9985 & 0.9995 \\ 
  $0.7$  & 0.8193 & 0.9897 & 0.9979 & 0.9992 && 0.8020 & 0.9947 & 0.9984 & 0.9994 \\ 
  $0.9$  & 0.8102 & 0.9889 & 0.9980 & 0.9993 && 0.8058 & 0.9925 & 0.9980 & 0.9993 \\ 
   \bottomrule
\end{tabular}
}
\label{TabBOpt}
\end{table}

\begin{table}[ht]
\caption{\footnotesize{Approximated Confidence for ($\d=0.1,\a=0.05$)-Bayesian tolerance interval}}
\centering
\raff{1.3}
\resizebox{\columnwidth}{!}{%
\begin{tabular}{cccccc} \toprule
  & \multicolumn{2}{c}{\emph{under Vanilla setup}} & \phantom{abc} 
  & \multicolumn{2}{c}{\emph{under Parameter Expansion setup}} \\
  \cmidrule{2-3} \cmidrule{5-6}
  & $A=\E(\n\given X)$ & $A=Optimal$ &
  & $A=\E(\n\given X)$ & $A=Optimal$ \\
  \midrule
  $intra-correlation$ \\
  $0.1$ & 0.972 & 0.969 && 0.968 & 0.963 \\ 
  $0.3$ & 0.964 & 0.955 && 0.955 & 0.949 \\ 
  $0.5$ & 0.936 & 0.921 && 0.925 & 0.917 \\ 
  $0.7$ & 0.925 & 0.907 && 0.915 & 0.911 \\ 
  $0.9$ & 0.952 & 0.941 && 0.940 & 0.935 \\ 
   \bottomrule
\end{tabular}
}
\label{TabRes}
\end{table}

\section{Frequentist justification of the Bayesian procedure}
\label{SectionFrequentistJustification}
In this section we show that Bayesian tolerance regions are often also approximate frequentist tolerance regions, of corresponding
levels. We consider an asymptotic setup, with data $X=X_n$ indexed by a parameter $n\ra\infty$,
in which the \emph{Bernstein-von Mises theorem} holds. The latter theorem (see e.g.\ \cite{vanderVaart(1998)}, Chapter~10) entails that
the posterior distribution $\Pi_n(\cdot\given X_n)$ of $\q$ can be approximated by
a normal distribution with deterministic covariance matrix, centered at an estimator $\hat\q_n=\hat\q_n(X_n)$, 
\begin{equation}
\label{EqBvM}
\Pi_n(\cdot\given X_n)-N\Bigl(\hat\q_n,\frac1n \Sigma_{\q}\Bigr)\prob 0,
\end{equation}
(in total variation norm), where the estimators $\hat\q_n=\hat\q_n(X_n)$ satisfy
\begin{equation}
\label{EqANMLE}
\sqrt n(\hat\q_n-\q)\generalweak\q N(0,\Sigma_\q).
\end{equation}
Under regularity conditions this is valid for $X_n$ a vector of $n$ i.i.d.\
observations from a smooth parametric model, with $\hat\q_n$ the 
maximum likelihood estimator and $\Sigma_\q$ the inverse Fisher information matrix.
More generally, this is true in the case of a local approximation by a Gaussian shift experiment
(\cite{LeCam72,vandervaartART}).
The Bernstein-von Mises theorem can be used to show that Bayesian and frequentist inference (testing and confidence sets)
\emph{merge} for large sample sizes. In this section we investigate this for tolerance intervals.

We shall show that Bayesian tolerance regions $\R_n(X_n)$ such that the functions 
\begin{equation}
\label{EqRescaledQ}
h\mapsto Q_{\hat\q_n(X_n)+h/\sqrt n}\bigl(\R_n(X_n)\bigr),\qquad n=1,2,\ldots,
\end{equation}
stabilize asymptotically to a \emph{deterministic} function are asymptotically frequentist
tolerance regions, for any given loss function $\ell$ and level $\a$. The crux of this
\emph{stability condition}
is that the randomness which enters the functions \eqref{EqRescaledQ} through $X_n$
in $\hat\q_n(X_n)$ asymptotically cancels the randomness which enters through $X_n$ in $\R_n(X_n)$: 
the Bayesian tolerance regions $\R_n(X_n)$ should 
be ``asymptotically pivotal'' with respect to the estimators $\hat\q_n$. 
Some type of stability condition appears to be necessary, 
because the shape of a Bayesian tolerance region is left free by its definition. 

An informal proof of the frequentist validity of Bayesian tolerance regions is as follows.
Replacing the posterior distribution in \eqref{EqBayesianellTolerance} by its normal
approximation \eqref{EqBvM} from the Bernstein-von Mises theorem, we find that
\begin{equation}
\label{EqAsDefTolEll}
\int \ell\bigl[Q_\vartheta\bigl(\R_n(X_n)\bigr)\bigr]\,dN\Bigl(\hat\q_n,\frac1n\Sigma_\q\Bigr)(\vartheta)\doteq
1-\a.
\end{equation}
By the  substitution $\vartheta=\hat\q_n+h/\sqrt n$ this can be rewritten in the form
\begin{equation}
\label{EqAsDefTolEllRescaled}
\int \ell\bigl[Q_{\hat\q_n+h/\sqrt n}\bigl(\R_n(X_n)\bigr)\bigr]\,dN\bigl(0,\Sigma_\q\bigr)(h)\doteq
1-\a.
\end{equation}
By the stability assumption the integrand 
\begin{equation}
\label{EqDefgn}
h\mapsto g_n(h;X_n):=\ell\bigl[Q_{\hat\q_n+h/\sqrt n}\bigl(\R_n(X_n)\bigr)\bigr]
\end{equation}
in this expression is asymptotically
the same as a deterministic function $h\mapsto g_\infty(h)$. In view of
\eqref{EqANMLE} the integral in \eqref{EqAsDefTolEllRescaled} is then approximately equal to 
$\E_\q g_\infty\bigl(\sqrt n(\q-\hat\q_n)\bigr)$, which in turn, again by stability, is asymptotically the same
as $\E_\q g_n\bigl(\sqrt n(\q-\hat\q_n); X_n\bigr)$, or
$$\E_\q \ell\bigl[Q_{\hat\q_n+\sqrt n(\q-\hat\q_n)/\sqrt n}\bigl(\R_n(X_n)\bigr)\bigr]
=\E_\q \ell\bigl[Q_{\q}\bigl(\R_n(X_n)\bigr)\bigr].$$ 
Thus the final expression, which is the frequentist level of the tolerance
region $\R_n(X_n)$, is asymptotically equal to $1-\a$.

For an $(\d,\a)$-tolerance region,  $\ell\bigl(Q_\q(\R_n(X_n)\bigr)$ is the
indicator of the set $\hat G_n=\{\q: Q_\q\bigl(\R_n(X_n)\bigr)\ge 1-\d\}$ and 
the function \eqref{EqDefgn} is the indicator of the set 
$$\hat H_n=\sqrt n(\hat G_n-\hat\q_n).$$
Thus the stability condition is that the latter sets approximate to a deterministic set,
as $n\ra\infty$. Condition \eqref{EqAsDefTolEllRescaled} becomes
\begin{equation}
\label{EqAsDefTolEllRescaledAs} 
N(0,\Sigma_\q)(\hat H_n)\doteq 1-\a.
\end{equation}
This equality allows to ``solve'' one aspect of the sets $\hat H_n$; in general additional constraints
will be imposed to define their shape. As the normal distribution in this display is fixed,
it is not unnatural that these constraints would render the sets $\hat H_n$ also to become fixed,
in the limit: stability is natural.

The following theorem makes the preceding rigorous. We shall verify its conditions for
normal prediction variables in the next section.

\begin{theorem}
Suppose that \eqref{EqBvM}--\eqref{EqANMLE} hold, the loss function $\ell$ is bounded,
and suppose that there exist (deterministic) functions $f_{n,1}, f_{n,2}:\RR^d\to \RR$ with the property
that $f_{n,i}(h_n)\ra f_\infty(h)$ for $i=1,2$ and some function $f_\infty$ and any sequence $h_n\ra h$ with limit $h$ in a set of
probability one under the normal distribution in \eqref{EqANMLE} and such that
\begin{equation}
\label{EqStabilitySandwich}
f_{n,1}(h)\le \ell\bigl[Q_{\hat\q_n+h/\sqrt n}\bigl(\R_n(X_n)\bigr)\bigr]\le f_{n,2}(h),\qquad h\in\RR^d.
\end{equation}
Then $\int \ell\bigl(Q_\q(\R_n(X_n))\bigr)\,d\Pi(\q\given X_n)\ra 1-\a\in(0,1)$ in probability implies that
$\E_\q\ell\bigl(Q_\q(\R_n(X_n))\bigr)\ra 1-\a$, as $n\ra\infty$, for every $\q$.
\end{theorem}

\begin{proof}
We may assume without loss of generality that the functions $f_{n,i}$ are uniformly bounded.
Then the condition $f_{n,i}(h_n)\ra f_\infty(h)$ for every sequence $h_n\ra h$ implies that
$\E f_{n,i}(Y_n)\ra\E f_\infty(Y)$, whenever the sequence of random vectors $Y_n$ tends in distribution
to the random vector $Y$, in view of the extended continuous mapping theorem (see \cite{vanderVaartandWellner(1996)},
Theorem~1.11.1). Thus by \eqref{EqANMLE}, for $i=1,2$,
$$\E_\q f_{n,i}\bigl(\sqrt n(\q-\hat\q_n)\bigr)\ra \int f_\infty\,dN(0,\Sigma_\q).$$
By assumption the function $g_n$ given in \eqref{EqDefgn} is sandwiched between
$f_{n,1}$ and $f_{n,2}$. Therefore
$\E_\q\ell\bigl(Q_\q(\R_n(X_n))\bigr)=\E_\q g_n\bigl(\sqrt n(\q-\hat\q_n); X_n\bigr)$
tends to the same limit.

By \eqref{EqBvM} and the definition of $g_n$ (see \eqref{EqAsDefTolEll}-\eqref{EqAsDefTolEllRescaled}), we have
$$\int \ell\bigl(Q_\q(\R_n(X_n))\bigr)\,d\Pi(\q\given X_n)=\int g_n(h; X_n)\,dN(0,\Sigma_\q)(h)+o_P(1).$$
Again by sandwiching of $g_n(h;X_n)$ this is asymptotic to 
$\int f_{n,i}\,dN(0,\Sigma_\q)$, and hence tends to $\int f_\infty\,dN(0,\Sigma_\q)$.
If the left side of the display tends to $1-\a$, as assumed, then it follows that
$\int f_\infty\,dN(0,\Sigma_\q)=1-\a$. The theorem follows by combining this with the preceding paragraph.
\end{proof}

\subsection{Normal predictions}
An $(\d,\a)$-tolerance interval $\R_n(X_n)=[A_n-B_n, A_n+B_n]$ for a one-dimensional Gaussian variable
$Z\sim N(\n,\t^2)$ is the  base (the section at $\t=0$) of a set of the form
$$G_{A,B,\d}=\Bigl\{\q=(\n,\t): \Phi\Bigl(\frac{A+B-\n}\t\Bigr)-\Phi\Bigl(\frac{A-B-\n}\t\Bigr)\ge 1-\d\Bigr\}.$$
The values $A_n$ and $B_n$ are determined so that $\Pi(G_{A_n,B_n,\d}\given X_n)\ge1-\a$,
for $\Pi(\cdot\given X_n)$ the posterior distribution of $\q=(\n,\t)$, and so that the length $2B_n$ of the interval
is minimal.

Under \eqref{EqBvM} the posterior distribution contracts  (at rate $1/\sqrt n$) to the Dirac measure at $\hat\q_n$, which tends to
the true value of the parameter $(\n,\t)$ under \eqref{EqANMLE}. Hence the equation
$\Pi(G_{A_n,B_n,\d}\given X_n)\ge1-\a$ forces that any ball of fixed radius around this true value intersects $G_{A_n,B_n,\d}$ with probability
tending to one. The minimality of $B_n$ implies that the horizontal locations $A_n$ of the latter sets must tend to the true value of $\n$. Thus 
$A_n-\hat\n_n\ra0$ in probability and hence
$\Phi(B_n/\hat\t_n)-\Phi(-B_n/\hat\t_n)\ra  1-\d$, whence $B_n/\hat\t_n\ra \xi_{\d/2}$, for $\xi_\d$ the upper $\d$-quantile of the
standard normal distribution.

The function $\q\mapsto \ell\bigl(Q_\q([A-B,A+B])\bigr)$ corresponding to the $(\d,\a)$-tolerance interval is the indicator of the set $G_{A,B,\d}$, 
and the stability condition \eqref{EqStabilitySandwich} is that the (indicator functions) of the sets $\hat H_n=\sqrt n(G_{A_n,B_n,\d}-\hat\q)$
are asymptotically deterministic. These sets can be written $\hat H_n=\{(g,h): K_n(g/\sqrt n, h/\sqrt n)\ge 0\}$,
for the stochastic processes
\begin{equation}
\label{EqDefKn}
K_n(g,h)=\Phi\Bigl(\frac{A_n+B_n-\hat\n_n-g}{\hat\t_n+h}\Bigr)-\Phi\Bigl(\frac{A_n-B_n-\hat\n_n-g}{\hat\t_n+h}\Bigr)- (1-\d).
\end{equation}
By a second-order Taylor expansion we see that these processes satisfy the expansion \eqref{EqExpansionK} in Lemma~\ref{LemmaExpansionK} (below),
with
\begin{align*}
\hat a_n=K_n(0,0)&=\Phi\Bigl(\frac{A_n+B_n-\hat\n_n}{\hat\t_n}\Bigr)-\Phi\Bigl(\frac{A_n-B_n-\hat\n_n}{\hat\t_n}\Bigr)- (1-\d),\\
\hat b_n=-\Bigl(\frac{\partial}{\partial h}K_n\Bigr)(0,0)
&=\psi\Bigl(\frac{A_n+B_n-\hat\n_n}{\hat\t_n}\Bigr)\frac1{\hat\t_n}-\psi\Bigl(\frac{A_n-B_n-\hat\n_n}{\hat\t_n}\Bigr)\frac1{\hat\t_n},\\
V(g,h)&=h.
\end{align*}
Here $\psi(x)=\phi(x)x=-\phi'(x)$. (Note that the partial derivatives
$-\phi((A_n+B_n-\hat\n_n)/\hat\t_n)/\hat\t_n+\phi((A_n-B_n-\hat\n_n)/\hat\t_n)/\hat\t_n$ of $K_n$ at $(0,0)$ relative to its first argument $g$ tend to zero.)
Since $A_n$, $B_n$, $\hat\n_n$ and $\hat\t_n$ tend in probability to nontrivial limits, the conditions of 
Lemma~\ref{LemmaExpansionK} are satisfied and hence the sets $\hat H_n$ are asymptotically sandwiched between
pairs of deterministic sets. Functions $f_{n,i}$ as in \eqref{EqStabilitySandwich} can be constructed from these sets
by letting $\e\ra 0$ and $M\ra\infty$ slowly with $n$. Thus the conditions of Theorem~\ref{EqStabilitySandwich} are satisfied 
and we obtain the following corollary.

\begin{corollary}
\label{CorollaryAsymptotic}
If the posterior distribution of $\q=(\n,\t)$ given $X_n$ satisfies \eqref{EqBvM}--\eqref{EqANMLE}, 
then the Bayesian $(\d,\a)$-tolerance interval $[A_n-B_n,A_n+B_n]$ of minimal length 
for a future variable $Z\given X_n,\n,\t\sim N(\n,\t^2)$ is an asymptotic frequentist $(\d,\a)$-tolerance set.
\end{corollary}

The convergence $A_n-\hat\n_n\ra 0$ in probability means that the tolerance intervals are asymptotically 
centered at the (asymptotic) posterior mean. If the posterior distribution of $\q$ is exactly normal $N\bigl((\hat\n_n,\hat\t_n),\Sigma_\q\bigr)$
with a diagonal covariance matrix, then $\E(\n\given X_n,\t)$ is free of $\t$ and Lemma~\ref{LemmaToleranceIntervalAtPosteriorMean}
shows that the tolerance interval is centered exactly at the posterior mean. 

For a non-diagonal matrix $\Sigma_\q$ this is not necessarily true, and in general the normal distribution will also be an approximation only.
The approximation $A_n-\hat\n_n\ra 0$ can in general be improved to order $n^{-1/4}$. The following lemma also gives an asymptotic expression for the
half length $B_n$ of the interval.

\begin{lemma}
The centers and half lengths of the Bayesian $(\d,\a)$-tolerance intervals $[A_n-B_n,A_n+B_n]$ of minimal length
in Corollary~\ref{CorollaryAsymptotic} satisfy $A_n=\hat\n_n+o_P(n^{-1/4})$ and 
$B_n=\hat\t_n\xi_{\d/2}+\xi_\a\xi_{\d/2}\sqrt{\Sigma_{\q,2,2}}n^{-1/2}+o_P(n^{-1/2})$.
\end{lemma}

\begin{proof}
Since $A_n-\hat\n_n\ra 0$ and $B_n\ra \xi_{\d/2}$, there exist intervals $I_n$ and $J_n$ around $\n_0$ and $\t_0\xi_{\d/2}$,
for $(\n_0,\t_0)$ the true value of $\q$, that shrink
to these points that contain $A_n$ and $B_n$ with probability tending to one. Define functions $F_n, G_n: I_n\times J_n\to \RR$ by
\begin{align*}
F_n(A,B)&=\Pi(G_{A,B,\d}\given X_n),\\
G_n(A,B)&=N(0,\Sigma_\q)\Bigl\{(g,h): h\le \frac {\sqrt n K_n(0,0;A,B)}{2\psi(\xi_{\d/2})/\hat\t_n}\Bigr\},
\end{align*}
Here $K_n(h,g;A,B)$ is the expression on the right side of \eqref{EqDefKn}, but  with $(A_n,B_n)$ replaced by a generic $(A,B)$,
and $\psi(x)=x\phi(x)$. The values $(A_n,B_n)$ are determined so that $B_n$ is the smallest value so that there exists 
$A_n$ such that $F_n(A_n,B_n)\ge 1-\a$. In other words $B_n$ is the minimum of $\{B: \sup_A F_n(A,B)\ge 1-\a\}$ and
$A_n$ is the point of maximum of $A\mapsto F_n(A,B_n)$. 
We shall show that the functions $F_n$ and $G_n$ satisfy the conditions of Lemma~\ref{LemmaCompareFnGn}
with $c_n=\hat\t_n\xi_{\d/2}$ and $\xi(\a)=\xi_\a\xi_{\d/2}\sqrt{\Sigma_{\q,2,2}}$, whence the lemma follows from that lemma.

In view of \eqref{EqBvM} we have $\sup_{A,B}\bigl|F_n(A,B)-N(0,\Sigma_\q)(\hat H_{A,B,\d})\bigr|\ra 0$, for
$\hat H_{A,B,\d}=\bigl\{(g,h): K_n(g/\sqrt n,h/\sqrt n;A,B)\ge0\bigr\}$. The supremum here and in the following is taken
over $(A,B)\in I_n\times J_n$.
By a second-order Taylor expansion, for functions $\h_{n,1}$ and $\h_{n,2}$ with $\sup_{A,B}\h_{n,i}(A,B)\ra 0$ and a constant $b$
independent of $(A,B)$,
\begin{align*}
&\sqrt n\biggl| K_n\Bigl(\frac g{\sqrt n},\frac h{\sqrt n};A,B\Bigr)-K_n(0,0;A,B)\\
&\qquad-\h_{n,1}(A,B)\frac g{\sqrt n}-\Bigl(-\frac{2\psi(\xi_{\d/2})}{\hat\t_n}+\h_{n,2}(A,B)\Bigr)\frac h{\sqrt n}\biggr|
\le b \frac{g^2+h^2}{\sqrt n}.
\end{align*}
By a sandwiching argument as in the proof of Lemma~\ref{LemmaExpansionK}, this shows that, with $C_n=\{(g,h): g^2+h^2\le L_n\}$ and $L_n\ra\infty$ sufficiently slowly
$$\sup_{A,B}\Bigl|N(0,\Sigma_\q)(\hat H_{A,B,\d}\cap C_n)
-N(0,\Sigma_\q)\Bigl\{(g,h)\in C_n: h\le \frac {\sqrt nK_n(0,0;A,B)}{2\psi(\xi_{\d/2})/\hat\t_n}\Bigr\}\Bigr|\ra 0.$$
Because $N(0,\Sigma_\q)( C_n^c)\ra 0$, this is then also true without intersecting the sets by $C_n$. This finishes
the proof that $\sup_{A,B}\bigl|F_n(A,B)-G_n(A,B)\bigr|\ra 0$.

By the unimodality and symmetry of the normal distribution, the map $A\mapsto K_n(0,0;A,B)$  has a maximum at $A=\hat\n_n$, for every $B$,
and hence the same is true for $A\mapsto G_n(A,B)$.
It is elementary to verify by several Taylor expansions that $G_n(\hat\n_n, \hat B)\ra 1-\a$, for 
$\sqrt n(\Phi(\hat B/\hat\t_n)-1+\d/2)\hat \t_n\ra \xi_\a\psi(\xi_{\d/2})\sqrt{\Sigma_{\q,2,2}}$, or 
$\hat B=\hat\t_n\xi_{\d/2}+\xi_\a\xi_{\d/2}\sqrt{\Sigma_{\q,2,2}} /\sqrt n$.

Finally, since $\frac{\partial}{\partial A} K_n(0,0;\hat\n_n,B)=0$,
again by Taylor expansion there exist functions $\h_{n,3}$ with the property $\sup_{A}\h_{n,3}(A,\hat B)\ra 0$ such that
$$K_n(0,0;A,\hat B)\le K_n(0,0;\hat\n_n,\hat B)+(A-\hat\n_n)^2\bigl(\phi'(\hat B/\hat\t_n)/\hat\t_n^2+\h_{n,3}(A,\hat B)\bigr).$$
As $\phi'(x)<0$, for $x>0$, we obtain with probability tending to one that $\sqrt n K_n(0,0;A,\hat B)$ is smaller than $\sqrt n K_n(0,0;\hat\n_n,\hat B)-\e$
for some $\e>0$ if $\sqrt n(A-\hat\n_n)^2>\d$ for some $\d>0$.
On this event $$G_n(A,\hat B)=N(0,\Sigma_\q)\Bigl\{(g,h): h\le \frac {\sqrt nK_n(0,0;A,\hat B)}{2\psi(\xi_{\d/2})/\hat\t_n}\Bigr\}$$
is strictly smaller than its asymptotic value $1-\a$ at $A=\hat\n_n$.
This verifies the last displayed condition of Lemma~\ref{LemmaCompareFnGn}.
\end{proof}

\begin{lemma}
\label{LemmaExpansionK}
Suppose that for every $M>0$ the stochastic processes $(K_n(h): h\in\RR^d)$ satisfy
\begin{equation}
\label{EqExpansionK}
\sup_{\|h\|\le M/\sqrt n}\sqrt n\bigl| K_n(h)-\hat a_n+\hat b_n Vh\bigr|\prob 0,
\end{equation}
for random variables $\hat a_n$ and $\hat b_n>0$ such that $\hat b_n^{-1}$
 is bounded in probability, and a linear map $V: \RR^d\to \RR$. 
If the sets $\hat H_n=\{h\in\RR^d: K_n(h/\sqrt n)\ge0\}$ satisfy
$N(0,\Sigma)(\hat H_n)\ra 1-\a\in(0,1)$, then  $\sqrt n\,\hat a_n/\hat b_n\ra \xi_\a \sqrt{V\Sigma V^T}$ and
for every $\e, M>0$, with probability tending to 1,
\begin{align*}
\bigl\{h\in \RR^d: \|h\|\le M,Vh&\le \xi_\a\sqrt{V\Sigma V^T}-\e\bigr\}\subset \hat H_n\subset\\
&\quad\subset \bigl\{h\in \RR^d: Vh\le \xi_\a\sqrt{V\Sigma V^T}+\e\text{ or } \|h\|>M\bigr\}.
\end{align*}
\end{lemma}

\begin{proof}
Define $\e_n(h)=\sqrt n\big(K_n(h)-\hat a_n+\hat b_n Vh\bigr)$ and set $\hat\e_n=\sup_{\|h\|\le M/\sqrt n}|\e_n(h)|$,
for given $M$.
Then by assumption $\hat\e_n\ra 0$ in probability, and $\bigl| K_n(h/\sqrt n)-\hat a_n+\hat b_n Vh/\sqrt n\bigr|\le \hat\e_n$,
for every $h$ with $\|h\|\le M$. From the latter inequality we find that 
\begin{align*}
\|h\|\le M, K_n(h/\sqrt n )\ge 0&\implies \hat b_n\, Vh\le \sqrt n\hat a_n+\hat\e_n,\\
\|h\|\le M,  \hat b_n\,Vh\le \sqrt n\hat a_n-\hat\e_n&\implies K_n(h/\sqrt n )\ge 0.
\end{align*}
This implies that
$$\bigl\{\|h\|\le M,Vh\le (\sqrt n\hat a_n-\hat\e_n)/\hat b_n\Bigr\}\subset \hat H_n
\subset \bigl\{Vh\le (\sqrt n\hat a_n+\hat\e_n)/\hat b_n\text{ or } \|h\|>M\bigr\}.$$
Combining this with the fact that $N(0,\Sigma)(\hat H_n)\ra 1-\a\in(0,1)$ and the fact that
$Vh\sim N(0, V\Sigma V^T)$ if $h\sim N(0,\Sigma)$, we conclude that
there exists $\d_M>0$ such that $\d_M\ra 0$ as $M\ra\infty$ such that
$(\sqrt n\hat a_n-\hat\e_n)/\hat b_n\le \xi_{\a-\d_M}\sqrt{V\Sigma V^T}+o_P(1)$ and
$(\sqrt n\hat a_n+\hat\e_n)/\hat b_n\ge \xi_{\a+\d_M}\sqrt{V\Sigma V^T}+o_P(1)$, for every $M$.
Since $\hat b_n^{-1}=O_P(1)$ by assumption, we have $\hat\e_n/\hat b_n\ra0$ in probability, and 
hence $\sqrt n\hat a_n/\hat b_n=\xi_\a\sqrt{V\Sigma V^T}+o_P(1)$.
We substitute this in the last display to obtain the result of the lemma.
\end{proof}

\begin{lemma}
\label{LemmaCompareFnGn}
Let $F_n,G_n: I_n\times J_n\to \RR$ be stochastic processes indexed by rectangles $I_n\times  J_n\subset\RR^2$ that are nondecreasing in their second argument,
such that $$\sup_{A,B}|F_n(A,B)-G_n(A,B)|\prob0,$$ and such that for numbers $c_n\in J_n$ and continuous functions $\xi: (0,1)\to\RR$,
and every $\a\in (0,1)$,
\begin{align*}
\sup_A G_n(A,c_n+\xi(\a)/\sqrt n)=G_n(0,c_n+\xi(\a)/\sqrt n)&\prob  1-\a,\\
\Pr\Bigl(\sup_{A: |A|>\d_nn^{-1/4}}G_n(A,c_n+\xi(\a)/\sqrt n) <  1-\a,\Bigr)&\ra 1,\quad \text{ some } \d_n\ra0.
\end{align*}
Then $B_n(\a):=\inf(B: \sup_A F_n(A,B)\ge 1-\a\bigr)$ satisfies $B_n(\a)=c_n+\xi(\a)/\sqrt n+o_P(n^{-1/2})$
and $\argmax_A F_n(A,B_n)=o_P(n^{-1/4})$.
\end{lemma}

\begin{proof}
The functions  $\bar F_n$ and $\bar G_n$ defined by $\bar F_n(B)=\sup_A F_n(A,B)$ and similarly for $G_n$
satisfy $\sup_B|\bar F_n(B)-\bar G_n(B)|\prob 0$. Combined with the first displayed assumption on $G_n$, this gives 
that $\bar F_n(c_n+\xi(\a)/\sqrt n)\prob 1-\a$, for every $\a\in(0,1)$. The definition of $B_n(\a)$ and monotonicity of $\bar F_n$
now readily give that $c_n+\xi(\a_2)/\sqrt n\le B_n(\a)\le c_n+\xi(\a_1)/\sqrt n$ for every $\a_1<\a<\a_2$, eventually with probability tending to one, 
or equivalently $\xi(\a_2)\le \sqrt n\bigl(B_n(\a)-c_n\bigr)\le \xi(\a_1)$, eventually. By the continuity of $\xi$ it follows that 
$\sqrt n\bigl(B_n(\a)-c_n\bigr)\prob \xi(\a)$.

By the uniform approximation of $F_n$ by $G_n$ and the second displayed assumption on $G_n$, we have that 
$$\sup_{|A|>\d_nn^{-1/4}}F_n\bigl(A,B_n(\a)\bigr)=\sup_{|A|>\d_n n^{-1/4}}G_n\bigl(A,B_n(\a)\bigr)+o_P(1).$$ 
By monotonicity $G_n(A,B_n(\a))\le G_n(A, c_n+\xi(\a_1)/\sqrt n)$, for every $\a_1<\a$, for every $A$. 
Thus the right side of the display is strictly smaller than $1-\a_1$, eventually,  by assumption, which is strictly
smaller than $1-\a$, for $\a_1$ close enough to $\a$. 
Similarly also $F_n\bigl(0,B_n(\a)\bigr)=G_n\bigl(0,B_n(\a)\bigr)+o_P(1)\ra 1-\a$. 
It follows that the maximum of $A\mapsto F_n\bigl(A,B_n(\a)\bigr)$ is taken on the interval $[-\d_n n^{-1/4}, \d_n n^{-1/4}]$.
\end{proof}

\bibliographystyle{plain}
\bibliography{BayesTI}

\end{document}